 \documentclass[UKenglish,cleveref, autoref, thm-restate]{lipics-v2021}
\usepackage{cite}
\usepackage{amsmath,amssymb,amsfonts}

\usepackage{algorithm}
\usepackage{algpseudocode}

\usepackage{graphicx}
\usepackage{textcomp}
\usepackage{xcolor}
\usepackage{url}
\usepackage{amsmath}
\usepackage{amssymb}
\usepackage{amsthm}
\usepackage{hyperref}
\usepackage{booktabs}

\usepackage{graphicx}
\usepackage{tikz}
\usetikzlibrary{cd}
\usepackage{thm-restate}

\newcommand{\ol}{\overline}
\newcommand{\distance}{\texttt{distance}}

\newcommand{\nb}{\overline{n}}

\hideLIPIcs
\nolinenumbers

\newtheorem{invariant}{Invariant}
  \usepackage{todonotes}
  \tikzset{notestyleraw/.append style={rectangle}}
\author{Ivor van der Hoog}{Technical University of Denmark, Denmark}{idjva@dtu.dk}{
https://orcid.org/0009-0006-2624-0231}{}

\author{Eva Rotenberg}{Technical University of Denmark, Denmark}{erot@dtu.dk}{
https://orcid.org/0000-0001-5853-7909}{}

\author{Daniel Rutschmann}{Technical University of Denmark, Denmark}{daru@dtu.dk}{
https://orcid.org/0009-0005-6838-2628}{}

\def\BibTeX{{\rm B\kern-.05em{\sc i\kern-.025em b}\kern-.08em
    T\kern-.1667em\lower.7ex\hbox{E}\kern-.125emX}}
\begin{document}

\title{Simpler Universally Optimal Dijkstra
}

\titlerunning{Simpler Universally Optimal Dijkstra}

\authorrunning{Ivor van der Hoog, Eva Rotenberg, and Daniel Rutschmann}

\Copyright{Ivor van der Hoog, Eva Rotenberg, and Daniel Rutschmann}

\funding{{\it Ivor van der Hoog}, {\it Eva Rotenberg}, and {\it Daniel Rutschmann} are grateful to the Carlsberg Foundation for supporting this research via Eva Rotenberg's Young Researcher Fellowship CF21-0302 ``Graph Algorithms with Geometric Applications''. This work was supported by the the VILLUM Foundation grant (VIL37507) ``Efficient Recomputations for Changeful Problems'' and the European Union's Horizon 2020 research and innovation programme under the Marie Sk\l{}odowska-Curie grant agreement No 899987. }



\Copyright{}

\maketitle

\begin{abstract}
Let $G$ be a weighted (directed) graph with $n$ vertices and $m$ edges. Given a source vertex $s$, Dijkstra’s algorithm computes the shortest path lengths from $s$ to all other vertices in $O(m + n \log n)$ time. This bound is known to be worst-case optimal via a reduction to sorting.
Theoretical computer science has developed numerous fine-grained frameworks for analyzing algorithmic performance beyond standard worst-case analysis, such as instance optimality and output sensitivity.
Haeupler, Hlad{\'\i}k,  Rozho{\v{n}}, Tarjan, and  T{\v{e}}tek [FOCS '24] consider the notion of universal optimality, a refined complexity measure that accounts for both the graph topology and the edge weights.
For a fixed graph topology, the universal running time of a weighted graph algorithm is defined as its worst-case running time over all possible edge weightings of $G$. An algorithm is universally optimal if no other algorithm achieves a better asymptotic universal running time on any particular graph topology.

Haeupler, Hlad{\'\i}k,  Rozho{\v{n}}, Tarjan, and  T{\v{e}}tek show that Dijkstra’s algorithm can be made universally optimal by replacing the heap with a custom data structure. Their approach builds on Iacono’s [SWAT '00] working-set bound $\phi(x)$. This is a technical definition that, intuitively, for a heap element $x$, counts the maximum number of simultaneously-present elements $y$ that were pushed onto the heap whilst $x$ was in the heap.
They design  a new heap data structure that can pop an element $x$ in $O(1 + \log \phi(x))$ time. 
They show that Dijkstra's algorithm with their heap data structure is universally optimal. 

In this work, we revisit their result. We introduce a simpler heap property called timestamp optimality, where the cost of popping an element $x$ is logarithmic in the number of elements inserted between pushing and popping $x$. We show that timestamp optimal heaps are not only easier to define but also easier to implement. Using these time stamps, we provide a significantly simpler proof that Dijkstra’s algorithm, with the right kind of heap, is universally optimal.
\end{abstract}

\newpage
\section{Introduction}

Let $G = (V, E)$ be a (di)graph with edge weights $w$, $n$ vertices, and $m$ edges. Let $s$ be a designated source vertex. 
In the \emph{single-source shortest path} problem, the goal is to sort all the vertices in $V$ by the length of the shortest path from $s$ to $v$. 
Dijkstra’s algorithm with a Fibonacci heap solves this problem in $O(m + n \log n)$ time, which is worst-case optimal.

However, worst-case optimality can be overly coarse as a performance measure. 
A stronger notion, called \emph{universal optimality}, requires an algorithm to run as fast as possible on \emph{every} graph topology.
Recently, Haeupler, Hlad{\'\i}k,  Rozho{\v{n}}, Tarjan, and  T{\v{e}}tek~\cite{Haeupler24} proved that Dijkstra’s algorithm, when paired with a suitably designed heap, achieves universal optimality. The paper was well-received, as it won the best paper award at FOCS'24. 
Their contribution is threefold:
First, they observe that Dijkstra’s efficiency relies on a heap supporting \textsc{Push} and \textsc{DecreaseKey} in amortized constant time, and \textsc{Pop} in $O(\log n)$ time. 
They refine this analysis using the \emph{working-set bound} by  Iacono~\cite{iacono2000improved} (we forward reference Definition~\ref{def:working_set_size}). Intuitively, for a heap element $x$, the working set size $\phi(x)$ is the maximum number of elements $y$ that are simultaneously present in the heap while $x$ resides in it (excluding all $y'$ that already were in the heap). 
Haeupler et al.~\cite{Haeupler24} design a new heap data structure for in the word-RAM model in which the \textsc{Pop} operation for element $x$ takes $O(1 + \log \phi(x))$ time.

 Secondly, they show, for any fixed graph topology $G = (V, E)$ and designated source $s$, a universal lower bound of $
\Omega(m + n + \sum_{x_i \in V \setminus \{ s \} } \log \phi(x_i))$. 
This lower bound matches the running time of Dijkstra's algorithm, when it is equipped with their heap. 
They thereby prove that Dijkstra's algorithm is universally optimal.

Third, they refine their analysis by distinguishing between edge weight comparisons and all other operations. 
Let $m_G \leq m$ be some topology-dependent parameter.
By an extensive algebraic argument, they prove a lower bound of $
\Omega( m_G + \sum_{x_i \in V \setminus \{ s \} } (1 + \log( \phi(x_i))))$
on the number of weight comparisons.
They adapt Dijkstra's algorithm so that the number of edge weight comparisons performed is also asymptotically universally optimal. 

A recent paper by Haeupler et al.~\cite{haeupler2025bidirectional} shows a simple instance optimal result for the bidirectional problem: Compute, given $G$ and two vertices $(s, t)$, the shortest $st$-path. Although the problems are related, the techniques and bounds are different, and thus, we do not consider the $st$-path problem in the paper at hand. 

\subparagraph{Contribution.}
We offer a significantly more concise proof of Dijkstra’s universal optimality.
Our main insight is to replace the working-set bound with a conceptually simpler measure called \emph{timestamps}.  
We maintain a heap $H$ and global time counter that increments after $H$.\textsc{Push}. 
For each heap element $x_i$, let $a_i$ and $b_i$ denote the time it was pushed and popped, respectively. 
We define $H$ to be \emph{timestamp optimal} if popping $x_i$ takes $O(1 + \log(b_i - a_i))$ time. 
Timestamps are not only conceptually simpler than working-set sizes, timestamp optimality is also considerably simpler to implement. 
We design a heap data structure that supports \textsc{Push} and \textsc{DecreaseKey} in amortized constant time, and \textsc{Pop} in amortized $O(1 + \log(b_i - a_i))$ time. 
Using timestamps, we show a universal lower bound of $
\Omega (m + \sum_{x_i \in V \setminus \{ s \} } (1 + \log (b_i - a_i)) )$ for any fixed graph topology $G = (V, E)$ and source $s$.  
We consider this to be our main contribution, as it offers a significantly simpler proof that Dijkstra's algorithm can be made universally optimal.
Finally, our formulation of timestamps allow us to apply recent techniques for universally optimal sorting~\cite{haeupler2025fast, Hoog2025Simpler} to show that Dijkstra's algorithm, with our timestamp optimal heap, matches our universal lower bound. 
Additionally, we show in Appendix~\ref{app:comparisons} that our algorithm is equally general to that of~\cite{Haeupler24}, as it can be made universally optimal with respect to edge weight comparisons.
While this proof is indeed new, different, and perhaps simpler, it is not self-contained; therefore, we relegate it to an appendix.

\section{Preliminaries}

Let $G = (V, E)$ be a (directed) graph and let $w$ assign to each edge a positive weight. 
For ease of notation, we say that $G$ has $n+1$ vertices and $m$ edges. 
In the \emph{single-source shortest path} problem (SSSP), the input consists of $G$, $w$, and a designated source vertex $s$. 
The goal is to order all $n$ vertices in $V \setminus \{ s \}$ by the length of the shortest path from $s$ to each vertex $v \in V$. 
We follow~\cite{Haeupler24} and assume that the designated source $s$ can reach every vertex in $V \setminus \{ s \}$ and that no two vertices are equidistant from $s$.

For weighted graph algorithms $A$ that solves SSSP, the running time $\textnormal{Runtime}(A, G, s, w)$ depends on the graph topology $G = (V, E)$, the designated source vertex $s$, and the edge weights $w$. 
Let $\mathbb{G}_n$ denote the set of all (directed) graph topologies $G = (V, E)$ with $|V| = n + 1$. 
For a fixed graph topology $G$, let $\mathbb{W}(G)$ denote the (infinite) set of all possible edge weightings of $G$.
Then, for a fixed input size $n$, the worst-case running time of an algorithm $A$ that solves the single-source-shortest path problem is a triple maximum:
\[
\textnormal{Worst-case}(A, n) := \max_{G \in \mathbb{G}_n} \, \max_{s \in V(G)} \, \max_{w \in \mathbb{W}(G)} \, \textnormal{Runtime}(A, G, s, w).
\]
\noindent
An algorithm $A$ is worst-case optimal if there exists a constant $c$ such that, for sufficiently large $n$, there exists no algorithm $A'$ with $\textnormal{Worst-case}(A, n) > c \cdot \textnormal{Worst-case}(A', n)$. 

Haeupler, Wajc, and Zuzic~\cite{haeupler2021universally} introduce a more fine-grained notion of algorithmic optimality for weighted graphs called \emph{universal optimality}. 
Universal optimality requires an algorithm to perform optimally for every fixed graph topology.
Formally, in the single-source-shortest path problem, it removes the outer two maximizations: For a fixed algorithm $A$, graph topology $G = (V, E)$, and source $s$, the universal running time is
\[
\textnormal{Universal}(A, G, s) := \max_{w \in \mathbb{W}(G)} \textnormal{Runtime}(A, G, s, w).
\]
An algorithm $A$ is \emph{universally optimal} if there exists a constant $c$ such that for every fixed topology $G$ and every $s \in V$, no algorithm $A'$ satisfies $
\textnormal{Universal}(A, G, s) > c \cdot \textnormal{Universal}(A', G, s)$.

\begin{algorithm}[t]
\begin{algorithmic}
\Require (Di)graph $G = (V, E)$, weighting $w$, designated source $s \in V$
\State $\distance \gets$ array with $n$ entries equal to $\infty$
\State H $\gets$ empty min-heap
\State $\distance[s] = 0$
\State H.\textsc{push}$(s, 0)$
\While{$H$ is not empty}
    \State u $\gets$ H.\textsc{Pop}()
    \State \textsc{Output}.append($u$)
    \For{Each edge $(u, v) \in E$}
        \If{$\distance[v]$ is $\infty$}
        \State $\distance[v]\gets\distance[u] + w((u,v))$
        \State H.\textsc{Push}$(v, \distance[v])$
        \ElsIf{$v$ is in $H$ and $\distance[v] > \distance[u] + w((u,v))$}
        \State $\distance[v]\gets\distance[u] + w((u, v))$
        \State H.\textsc{DecreaseKey}$(v, \distance[v])$
        \EndIf
    \EndFor
\EndWhile
\end{algorithmic}
\caption{Dijkstra's algorithm with a black-box heap $H$.}
\label{algo:dijkstra}
\end{algorithm}

\subsection{Dijkstra's algorithm}

Dijkstra~\cite{Dijkstra1959} presented an algorithm for the single-source shortest paths problem in 1959.
Since the invention of Fibonnacci heaps\cite{fredman1987fibonacci}, Dijkstra's algorithm can be implemented in a way that is worst-case optimal, by utilising these heaps in the implementation of Dijkstra's algorithm.
A heap $H$ maintains a set of elements where each $x \in H$ is associated with a unique \emph{key} $\gamma(x)$. 
A heap $H$ typically supports the following five operations: 

\begin{itemize}
    \item $\textsc{Push}(x, y)$: Insert $x$ into $H$ where the key $\gamma(x)$ is $y$.
    \item $\textsc{DecreaseKey}(x, y)$: For $x \in H$ and $y \leq \gamma(x)$, update $\gamma(x)$ to $y$.
    \item $\textsc{Pop}()$: Remove and return the element $x \in H$ with the smallest key.
    \item $\textsc{Peak()}$: Return an element $x \in H$ with the smallest key.
    \item $\textsc{Meld}(H')$: Return a heap $H''$ on the elements in $H \cup H'$.
\end{itemize}

Fibonacci heaps support $\textsc{Push}$ and $\textsc{DecreaseKey}$ in $O(1)$ amortized time, $\textsc{Pop}$ in amortized $O(\log n)$ time, and \textsc{Peak} and \textsc{Meld} in worst-case $O(1)$ time~\cite{fredman1987fibonacci}. 
Dijkstra's algorithm only uses the first three heap operations. Using Fibonacci heaps, it runs in $O(m + n \log n)$ time which is worst-case optimal.
Dijkstra's algorithm (Algorithm~\ref{algo:dijkstra}) is conceptually simple: it initializes the heap with all vertices $v$ that $s$ can directly reach, where the key $\gamma(v)$ is the edge weight of $(s, v)$. Then, it  repeatedly pops the vertex $u$ with the smallest key and it appends $u$ to the output. 
Each time a vertex $u$ is popped, its neighbors $v$ are examined. For each neighbor $v$, the tentative distance $d_v$ is computed as the distance from $s$ to $u$ plus the edge weight $w(u, v)$. 
If $v$ is not in the heap, it is pushed with key $\gamma(v) = d_v$. 
If $v$ is already in the heap and $\gamma(v) > d_v$ then the key $\gamma(v)$ is decreased to $d_v$.

\subsection{Universally optimal Dijkstra}

Iacono~\cite{iacono2000improved} introduced a refined measure for the running time of heap operations via the concept of \emph{working-set size}:

\begin{definition}[\cite{iacono2000improved}~]
\label{def:working_set_size}
Consider a heap supporting \textsc{Push} and \textsc{Pop} and let $x_i$ be an element in the heap. 
For each time step $t$ between pushing $x_i$ and popping $x_i$, we define the working set of $x_i$ at time $t$ as the
set of elements $W(x_i,t)$ inserted after $x_i$ and still present at
time $t$. We include $x_i$ itself in $W(x_i,t)$.
Fix any time $t_0$ that maximizes the value of $|W(x_i,t)|$;
we call the set $W(x_i,t_0)$ \emph{working-set} of $x_i$ and $\phi(x_i) := |W(x_i,t_0)|$ the \emph{working-set size}.
\end{definition}

Haeupler, Hlad{\'\i}k,  Rozho{\v{n}}, Tarjan, and  T{\v{e}}tek ~\cite{Haeupler24} design a heap $H$ that supports \textsc{Push} and \textsc{DecreaseKey} in amortized $O(1)$ time, and supports \textsc{Pop} for an element $x \in H$ in $O(1 + \log \phi(x))$ time.
The basis of their data structure are Fibonacci heaps $F$ which support \textsc{Push} and \textsc{DecreaseKey} in  amortized constant time and \textsc{Pop} in amortized $O(\log N)$ time where $N$ is the number of elements in the heap. 

The core idea of their data structure is to partition elements into buckets $B_j$ of doubly exponentially increasing size. 
Each bucket $B_j$ contains a Fibonacci heap $F_j$. Denote by $m_j$ the minimum key of $F_j$. The set of $O(\log \log n)$ values $m_j$ is stored in a fusion tree~\cite{fredman1993surpassing}. Fusion trees of size at most $O(\log \log n)$ have constant update time in the word-RAM model.  
By carefully merging and managing buckets, they ensure that each element $x_i$ is stored in a Fibonacci heap $F_j$ whose size is proportional to $\phi(x_i)$. 
We refer to~\cite{Haeupler24} for full details; we only note here that ensuring that  \textsc{Push} and \textsc{DecreaseKey} take amortized constant time whilst \textsc{Pop} takes $O(1 + \log \phi(x_i))$ time requires non-trivial maintenance and analysis.

Using this heap, Dijkstra’s algorithm takes $O(m + \sum_{x_i \in V \setminus \{ s \} } (1 + \log \phi(x_i)))$ time. 
They~\cite{Haeupler24} also prove a universal lower bound: for any fixed (directed) graph topology $G = (V, E)$ and source $s \in V$, no algorithm $A'$ achieves runtime $o(m + n + \sum_{x_i \in V \setminus \{s \} } \log \phi(x_i))$ for all edge weightings $w \in \mathbb{W}(G)$. 
Dijkstra’s algorithm with this heap is therefore universally optimal.

\section{Timestamp optimal heaps}

Instead of designing heaps with a running time proportional to the working-set size of a heap element, we use a simple timestamp approach:

\begin{definition}
    For any heap $H$, we define \emph{time} $t$ as a counter that increments after each  $H.\textsc{Push}(x, y)$. 
    For any element $x_i \in H$, we denote by $a_i$ the time when $x_i$ was pushed onto the heap, and by $b_i$ the time when $x_i$ was popped from the heap. Observe that $b_i > a_i$.
\end{definition}

\begin{definition}
    A heap $H$ is \emph{timestamp optimal} if it supports \textsc{Push} and \textsc{DecreaseKey} in amortized constant time, and \textsc{Pop} with output $x_i$ in amortized $O(1 + \log (b_i - a_i))$ time.
\end{definition}

\subparagraph{Our Data Structure}
Let $H$ denote our abstract heap data structure, supporting operations $H.\textsc{Push}(x, y)$, $H.\textsc{DecreaseKey}(x, y)$, and $H.\textsc{Pop}()$. We implement $H$ as follows: 

\begin{invariant}[Figure~\ref{fig:invariant}]
\label{inv:buckets}
    We partition the elements of $H$ into an array of buckets $B$.
    Each bucket $B[j]$ contains one or two Fibonacci heaps. Every element of $H$ lies in exactly one Fibonacci heap. 
    Along with each Fibonacci heap $F$ we store a half-open interval $I_F \subset \mathbb{R}$. We maintain the following invariants; for all times $t$:
    \begin{enumerate}[(a)]
        \item An element $x_i \in H$ is in $F$ if and only if $a_i \in I_F$,\label{item:one}
        \item The interval $I_F$ for $F \in B[j]$ has size $2^j$, \label{item:two}
        \item All intervals $F'$ of $B[j-1]$ lie to the right of the intervals in $F$ in $B[j]$, \label{item:three}
        \item All intervals collectively partition $[0, t)$. \label{item:four}
    \end{enumerate}
\end{invariant}

\begin{figure}[b]
    \centering
    \includegraphics[width = \linewidth]{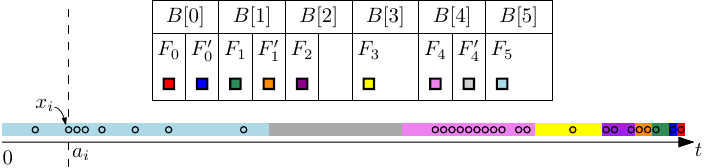}
    \caption{
    We illustrate our invariant. Each bucket $B[j]$ contains one or two Fibonacci heaps. Each heap $F$ is associated with an interval on the number line in $[0, t)$. 
The elements $x_i$ stored in a Fibonacci heap $F$ are visualised by placing them at positions $a_i$ along the number line. Importantly, when the buckets in $B$ are indexed from left to right (i.e., from low to high indices), the corresponding intervals of the Fibonacci heaps appear in the reverse order along the number line.
    }
    \label{fig:invariant}
\end{figure}

\noindent
By Invariant~\ref{inv:buckets}, for all elements $x_i \in B[j]$, $(t - a_i)$ is at least $2^{j-1}$ and at most $2^{j+1}$. We define $M$ as an array of size $\lceil \log n \rceil$, where $M[j]$ stores the minimum key in $B[j]$. 
Furthermore, we maintain a bitstring $S_M$ where $S_M[j] = 1$ if for all $k > j$, we have $M[j] \leq M[k]$. 

\begin{lemma}
    \label{lemm:get_bucket}
    For any element $x_i \in H$ we can, given $a_i$ and $t$, compute the bucket containing $x_i$ and the Fibonacci heap that contains $x_i$ in constant time.
\end{lemma}

\begin{proof} Let $a_i$ be the insertion time of $x_i$. We now want to determine the possible values $j$ such that $B[j]$ contains $x_i$. From invariant 1.b we get: $t - a_i \ge 2^{j-1}$, or equivalently, $j \le \log(t - a_i) + 1$; and $t - a_i \le 2^{j+1}$, or equivalently, $j \ge \log(t - a_i) - 1$.
In particular, for $\ell = \lfloor \log(t - a_i) \rfloor$, the Fibonacci heap containing $x_i$ lies in one of the buckets $B[\ell-1], B[\ell], B[\ell+1]$.
For each of these at most $6$ corresponding intervals of the at most $6$ heaps contained in those $3$ buckets, we test, in constant time,  whether $a_i$ is contained in the interval. 
This procedure returns us the interval containing $a_i$, and thus, the heap (and bucket) containing $x_i$.
%
\end{proof}

\begin{theorem}
    The above data structure is a timestamp optimal heap $H$.
\end{theorem}

\begin{proof}
    We prove that the data structure supports all three heap operations within the specified time bounds while maintaining Invariant~\ref{inv:buckets}. 

\subparagraph{Push($x_i$, $y$).} To push $x_i$ onto $H$, we set $a_i = t$ and increment $t$.
    We then add a new Fibonacci heap $F$ to $B[0]$ with $I_F = [a_i, a_i+1)$. 
    Whilst a bucket $B[j]$ contains three Fibonacci heaps, we select its oldest two Fibonacci heaps $F_1, F_2$.
    We \textsc{Meld} $F_1$ and $F_2$ into a new Fibonacci heap $F'$ in worst-case constant time, and we move $F'$ to $B[j+1]$. We define the corresponding interval $I_{F'}$ as the union of the consecutive intervals $I_{F_1}$ and $I_{F_2}$. Since Property~(\ref{item:one}) held for both $F_1$ and $F_2$, it now holds for $F'$.
    Moreover, $|I_{F'}| = 2^{j+1}$ and so Property~(\ref{item:two}) is maintained.
    Since, before the update, all intervals collectively partitioned $[0, t)$ and we only added the interval $[t, t+1)$ we maintain Properties~(\ref{item:three}) and (\ref{item:four}). 
    Every \textsc{Meld} operation decreases the number of Fibonacci heaps
    and every push operation creates only one heap.
    Hence, each push operation performs amortized $O(1)$ \textsc{Meld} operations.

    Finally, we update $M$ and the string $S_M$. 
    After we have merged two heaps in a bucket $B[j]$, only $M[j]$ and $M[j+1]$ change. Moreover, $M[j]$ can only increase in value and $M[j+1]$ can only decrease to at least the original value of $M[j]$. 
    Therefore, all bits $S_M[k]$ for $k \not \in \{ j, j+1 \}$ remain unchanged and we update $S_M$ in amortized constant time. 
    Thus, a push operation takes amortized constant time.
    
\subparagraph{DecreaseKey($x_i$, $y$).}
    By Lemma~\ref{lemm:get_bucket}, we find the Fibonacci heap $F$ and the bucket $B[j]$ that contain  $x_i$ in  constant time. We perform $F$.\textsc{DecreaseKey}$(x_i, y)$ in amortized constant time. 
    We update $M[j]$ and identify the index $k$ of the leftmost $1$-bit  in $S_M$ after index $j$ in constant time using bitstring operations.
    Note that $S_M[j] = 1$ if and only if $M[j] \leq M[k]$.

    If $S_M[j] = 1$, then we recursively find the next $j' < j$ for which $S_M[j'] = 1$ using the same bitstring operation. 
    If $M[j] < M[j']$, then we set $S_M[j'] = 0$ and we recurse.
    If $M[j] \geq M[j']$ then for all $j'' < j'$ with  $S_M[j''] = 1$, $M[j] \geq M[j'']$ and our update terminates. 
    The recursive call takes amortized constant time since it can only decrease the number of $1$-bits in $S_M$.

\subparagraph{Pop().}
    We obtain the index $j$ of the leftmost $1$-bit of $S_M$ in constant time using bitstring operations. 
    Per definition, the minimum key of $H$ is in $B[j]$. 
    We perform the \textsc{Peak} operation on both heaps in $B[j]$. For the heap $F$ containing the minimum key, we do $F$.\textsc{Pop()} in $O(1 + \log 2^j) = O(j)$ time to return the element $x_i \in H$ with the minimum key. 
    Finally, we update $M[j]$ and $S_M$. We update $M[j]$ in constant time and iterate from $j$ to $1$. For each iteration $\ell$, let $k > \ell$ be the minimum index where $S_M[k] = 1$. 
    We obtain $k$ in constant time using bitstring operations and update $S_M[\ell]$ by comparing $M[\ell]$ to $M[k]$.  Thus, this update takes $O(j)$ time. 
    By Invariant~\ref{inv:buckets}, $\log (b_i - a_i) \in O(j)$ and thus the $H$.\textsc{Pop()} takes $O(1 + \log (b_i - a_i) )$ amortized time. 
\end{proof}

\begin{corollary}
\label{cor:runtime}
Algorithm~\ref{algo:dijkstra} can be adapted to use $O(m + \sum\limits_{x_i \in V \setminus \{s \} } (1 + \log (b_i - a_i))$ time. 
\end{corollary}

\section{Universal Optimality}
\label{sec:universal}

We show that categorizing the running time of Dijkstra's algorithm by timestamps substantially simplifies its proof for universal optimality.
For a given graph $G = (V, E)$ and a designated source vertex $s$, we define a \emph{linearization} $L$ of $(G, s)$ as any permutation of $V \setminus \{s\}$ such that there exists some choice of edge weights $w' \in \mathbb{W}(G)$ where $L$ is the solution to the single-source-shortest path problem on $(G, s, w')$.
Let $\ell(G, s)$ denote the number of distinct linearizations of $(G, s)$.
Observe that $\Omega(\log \ell(G, s))$ is a comparison-based universal lower bound for the single-source shortest path problem.
Indeed, any comparison-based algorithm has a corresponding binary decision tree with leaves representing possible outputs.
With $\Omega(\ell(G, s))$ distinct outputs, there must exist a root-to-leaf path of length $\Omega(\log \ell(G, s))$.

\subparagraph{Our approach.}
We fix the input $(G = (V, E), s, w)$ and execute Dijkstra's algorithm, which defines a set of $n$ timestamp intervals $\{ [a_i, b_i] \}_{i \in [n]}$. Each interval $[a_i, b_i]$ has a minimum size of $1$ and corresponds to a unique vertex $x_i \in V \setminus \{s\}$. Dijkstra's algorithm operates in $O(m + \sum_{x_i \in V \setminus \{s\}} (1 + \log (b_i - a_i)))$ time. Since any algorithm must spend $\Omega(n + m)$ time to read the input, it remains to show that $\sum\limits_{x_i \in V \setminus \{s\}} (1  + \log (b_i - a_i)) \in \Omega(n + \log \ell(G, s))$.

\begin{lemma} \label{lemm:points_define_linearizations}
    For each vertex $x_i \in V \setminus \{s\}$, choose a real value $r_i \in [a_i, b_i]$. 
    If all $r_i$ are distinct, then the sequence $L$, which orders all $x_i$ by $r_i$, is a linearization of $(G, s)$. 
\end{lemma}
\begin{proof}
  Consider the initial execution of Dijkstra's algorithm on $(G, s, w)$. For each vertex $x_i \in V \setminus \{s\}$, there is a unique edge $(u, x_i)$ that, upon inspection, pushes the element $x_i$ onto the heap $H$. The set of these edges forms a spanning tree $T$ rooted at $s$. 

 Denote by $\textnormal{rank}(r_i)$ the rank of $r_i$ in the ordered sequence $L$. 
We construct a new integer edge weighting $w'$ such that running Dijkstra's algorithm on $(G, s, w')$ outputs $L$. We set $w'(e) = \infty$ for all edges $e \in E \setminus T$. 
We then traverse $T$ in a breath-first manner.
For each edge $(x_j, x_i) \in T$, we observe that in our original run of Dijkstra's algorithm, the algorithm popped the vertex $x_j$ before it pushed the vertex $x_i$ onto the heap. It follows that for the intervals $\{ [a_j, b_j], [a_i, b_i] \}$, the value $b_j \leq a_i$. 
Since $r_j \in [a_j, b_j]$, $r_i \in [a_i, b_i]$, and $r_j \neq r_i$ we may conclude that $\textnormal{rank}(r_i) > (r_j)$. 
We assign to each edge $(x_j, x_i) \in T$ the positive edge weight $w'((x_j, x_i)) = \textnormal{rank}(r_i) - \textnormal{rank}(r_j)$. 
It follows that the distance in $(G, E, w')$ from $s$ to any vertex $x_i$ is equal to $\textnormal{rank}(r_i)$ and so Dijkstra's algorithm on $(G, s, w')$ outputs $L$. 
\end{proof}

What remains is to show that the number of ways to select distinct $r_i \in [a_i, b_i]$, correlates with the size of the timestamp intervals. This result was first demonstrated by Cardinal, Fiorini, Joret, Jungers, and Munro~\cite[the first 2 paragraphs of page 8 of the ArXiV version]{cardinal2010sorting}. It was later paraphrased in Lemma 4.3 in~\cite{haeupler2025fast}.

\begin{lemma} [Originally in~\cite{cardinal2010sorting}, paraphrased by Lemma 4.3 in~\cite{haeupler2025fast}] \label{lemm:intervals_to_linearizations}
    Let $\{ [a_i, b_i ] \}$ be a set of $n$ integer intervals where each interval has size at least $1$.   
    Let $Z$ be the set of linear orders realizable by real numbers $r_i \in [a_i, b_i]$.
    Then $ \sum\limits_{i \in [n]} \log( b_i - a_i) \in O(n + \log |Z|)$.
\end{lemma}

\begin{proof}
    To illustrate the simplicity of this lemma, we quote the proof in \cite{haeupler2025fast}:
    For each $i$ between $1$ and $n$ inclusive, choose a real number $r_i$ uniformly at random from the real interval $[0, n]$, independently for each $i$. With probability $1$, the $r_i$ are distinct. 
    Let $L$ be the permutation of $[n]$ obtained by sorting $[n]$ by $r_i$. Each possible permutation is equally likely.
    If each $r_i \in [a_i, b_i]$ then $L$ is in $Z$.
    The probability of this happening is $\prod_{i=1}^k (b_i-a_i)/n$.  It follows that $|Z| \geq n! \cdot \prod_{i=1}^n (b_i-a_i)/n$.  Taking logarithms gives $\log |Z| \geq \sum_{i=1}^n \log (b_i-a_i) + \log n! - n\log n$.  By Stirling's approximation of the factorial, $\log n! \geq n\log n - n\log e$.  The lemma follows.   
\end{proof}

It follows that we can recreate the main theorem from~\cite{Haeupler24}:

\begin{theorem}[Theorem~1.2 in~\cite{Haeupler24}]
    Dijkstra's algorithm, when using a sufficiently efficient heap data structure, has a universally optimal running time. 
\end{theorem}
\begin{proof}
    The running time of Dijkstra's algorithm with the input $(G, s, w)$ is $O(m)$ plus the time needed for push and pop operations.
    With a timestamp-optimal heap, by Corollary~\ref{cor:runtime}, the algorithm thus takes $O(m + \sum\limits_{x_i \in V \setminus \{s \} } (1 + \log (b_i - a_i)))$ total time. 
    By \cref{lemm:intervals_to_linearizations}, the total running time is then $O(m + n + \log |Z|)$, where $|Z|$ denotes the number of distinct linear orders created by choosing distinct real values $r_i \in [a_i, b_i]$. 
    By Lemma~\ref{lemm:points_define_linearizations}: $|Z| \leq \ell(G, s)$, and so  Dijkstra's algorithm with any timestamp-optimal heap runs in $O(m + n + \log \ell(G, s))$ time.
    Conversely, any algorithm that solves SSSP needs $\Omega(m + n)$ time to read the input,
    and $\Omega(\log \ell(G, s))$ is an information-theoretical universal lower bound.
\end{proof}

\section{Conclusion}
We study Dijkstra's algorithm for the single-source-shortest path problem and give an alternative proof that Dijkstra's algorithm is universally optimal. 
We consider our construction to be considerably simpler than the one in~\cite{Haeupler24}.
We find time stamps easier to define than the working set size of an element, and regard our timestamp optimal heap to be simpler in both its construction and analysis than the heap used in~\cite{haeupler2025fast}. As an added benefit, we only rely on Fibonacci trees and bitstring manipulation and no involved data structures (e.g., the \emph{fusion trees} used in~\cite{Haeupler24}). 
We regard Section~\ref{sec:universal} as our main contribution, as it gives a very concise proof of universal optimality. 
We understand that algorithmic simplicity is subjective, and we leave it to the reader to judge the simplicity of our approach.
We note, however, that proofs are unarguably shorter since the full version of~\cite{Haeupler24} counts over 50 pages. 
We observe that both our paper and~\cite{Haeupler24} must assume a word RAM and we consider it an interesting open problem to show universal optimality on a pointer machine. 

Finally, we note that our data structure is equally general to~\cite{Haeupler24}. In Appendix~\ref{app:comparisons}, we show that our data structure, just as the one in~\cite{Haeupler24}, is also universally optimal if algorithmic running time is exclusively measured by the number of edge weight comparisons that the algorithm performs. 
While our proof is indeed new, different, and perhaps simpler than the previous one, it is not self-contained; for this reason, we relegate it to an appendix.

\bibliographystyle{plainurl}
\bibliography{refs}

@inproceedings{Haeupler24,
  author       = {Haeupler, Bernhard and Hlad{\'\i}k, Richard and Rozho{\v{n}}, V{\'a}clav and Tarjan, Robert and T{\v{e}}tek, Jakub},
  title        = {Universal Optimality of Dijkstra via Beyond-Worst-Case Heaps},
  booktitle    = {Symposium on Foundations of Computer Science ({FOCS})},
  publisher    = {{IEEE}},
  year         = {2024},
  doi          = {10.1109/FOCS61266.2024.00125}
}

@inproceedings{haeupler2025fast,
  title={Fast and simple sorting using partial information},
  author={Haeupler, Bernhard and Hlad{\'\i}k, Richard and Iacono, John and Rozho{\v{n}}, V{\'a}clav and Tarjan, Robert E and T{\v{e}}tek, Jakub},
  booktitle={ACM-SIAM Symposium on Discrete Algorithms (SODA)},
  pages={3953--3973},
  year={2025},
  organization={SIAM},
    doi = {10.1137/1.9781611978322.134}
}

@inproceedings{cardinal2010sorting,
  title={Sorting under partial information (without the ellipsoid algorithm)},
  author={Cardinal, Jean and Fiorini, Samuel and Joret, Gwena{\"e}l and Jungers, Rapha{\"e}l M and Munro, J Ian},
  booktitle={ACM Symposium on Theory of Computing (STOC)},
  pages={359--368},
  year={2010},
    doi = {10.1145/1806689.1806740}
}

@article{fredman1993surpassing,
  title={Surpassing the information theoretic bound with fusion trees},
  author={Fredman, Michael and Willard, Dan},
  journal={Journal of computer and system sciences},
  volume={47},
  number={3},
  pages={424--436},
  year={1993},
  publisher={Elsevier},
    doi = {10.1016/0022-0000(93)90040-4}
}

@inproceedings{iacono2000improved,
  title={Improved upper bounds for pairing heaps},
  author={Iacono, John},
  booktitle={Scandinavian Workshop on Algorithm Theory},
  pages={32--45},
  year={2000},
  organization={Springer},
    doi = {10.5555/645900.672600}
}

@inproceedings{Hoog2019Preprocessing,
  title={Preprocessing Ambiguous Imprecise Points  },
  author={van der Hoog, Ivor  and Kostitsyna, Irina  and L\"{o}ffler, Maarten  and Speckmann, Bettina},
  booktitle={Symposium on Computational Geometry (SOCG) },
  year={2019},
doi = {10.4230/LIPIcs.SoCG.2019.42}
}

@inproceedings{Hoog2025Simpler,
  title={Simpler Optimal Sorting from a Directed Acyclic Graph},
  author={van der Hoog, Ivor and Rotenberg, Eva and Rutschmann, Daniel},
  booktitle={SIAM Symposium on Simplicity in Algorithms (SOSA)},
  year={2025},
doi = {10.1137/1.9781611978315.26}
}

@inproceedings{haeupler2021universally,
  title={Universally-optimal distributed algorithms for known topologies},
  author={Haeupler, Bernhard and Wajc, David and Zuzic, Goran},
  booktitle={ACM Symposium on Theory of Computing (STOC)},
  pages={1166--1179},
  year={2021},
    doi = {10.1145/3406325.3451081}
}

@article{Dijkstra1959,
  author    = {Edsger Dijkstra},
  title     = {A Note on Two Problems in Connexion with Graphs},
  journal   = {Numerische Mathematik},
  volume    = {1},
  pages     = {269--271},
  year      = {1959},
  doi       = {10.1007/BF01386390}
}

@inproceedings{haeupler2025bidirectional,
  title={Bidirectional Dijkstra’s Algorithm is Instance-Optimal},
  author={Haeupler, Bernhard and Hlad{\'\i}k, Richard and Rozho{\v{n}}, V{\'a}clav and Tarjan, Robert E and T{\v{e}}tek, Jakub},
  booktitle={Symposium on Simplicity in Algorithms (SOSA)},
  pages={202--215},
  year={2025},
  organization={SIAM},
    doi = {10.1137/1.9781611978315.16}
}

@article{fredman1987fibonacci,
  title={Fibonacci heaps and their uses in improved network optimization algorithms},
  author={Fredman, Michael and Tarjan, Robert},
  journal={Journal of the ACM (JACM)},
  volume={34},
  number={3},
  pages={596--615},
  year={1987},
  publisher={ACM New York, NY, USA},
    doi = {10.1145/28869.28874}
}

\appendix

\section{Universal optimality with regards to number of comparisons}
\label{app:comparisons}

Haeupler, Hlad{\'\i}k,  Rozho{\v{n}}, Tarjan, and  T{\v{e}}tek~\cite{Haeupler24} propose an alternative way to measure algorithmic running time where one counts \emph{only} the edge weight comparisons made by the algorithm (all other algorithmic instructions are considered `free'). 
They prove that Dijkstra's algorithm is not universally optimal when algorithmic running time is defined by this measure. 
They then modify Dijkstra's algorithm to obtain an algorithm that \emph{is} universally optimal under this runtime measure. 

Here, we prove that our data structure and algorithmic analysis are equally general when compared to those of~\cite{Haeupler24}: we prove that Dijkstra's algorithm under our heap can also be adapted to become universally optimal under this runtime measure.
To this end, we fix the input $(G, s, w)$ of our algorithm and define two additional key concepts:

\begin{definition}
      Consider Dijkstra's algorithm on $(G, s, w)$. For each vertex $x_i \in V \setminus \{s\}$, there is a unique edge $(u, x_i)$ that, upon inspection, pushes the element $x_i$ onto the heap $H$. We define $T$ as the (directed) spanning tree, rooted at $s$, that is formed by these edges. 
\end{definition}

\begin{definition}
    A (directed) edge $(u, v)$ of $G$ is a \emph{forward edge} if
$\distance(s, v) > \distance(s, u)$ and $(v, u)$ is not in $T$. We define $m_f$ as the number of forward edges.
\end{definition}


\noindent
We modify the algorithm to use $O(m_f + \sum\limits_{x_i \in V \setminus \{ s \}} (1 + \log (b_i - a_i)))$ edge weight comparisons.

\subparagraph{Bottleneck compression.} Consider the unweighted (di)graph $G$ and a breath-first search procedure on $G$ from $s$.
Let $d$ be one plus the maximum hop-distance from $s$ to any other vertex. 
We denote by $L_0, \dots, L_{d-1}$ the breath-first-search (BFS) levels of $G$.
That is, $L_i$ contains the vertices whose hop-distance distance from $s$ is $i$.
If $L_i$ consists of a single vertex, then we call that vertex a \emph{bottleneck vertex}.
In Dijkstra's algorithm, regardless of the choice in edge weights $w$, the bottleneck vertices always get processed in order. That is,
the $i$th bottleneck vertex gets pushed after the $(i-1)$st bottleneck vertex gets popped,
A \emph{bottleneck path} is a maximal sequence $u_1, \dots, u_k$
of bottleneck vertices that occupy consecutive levels.

\begin{observation}
    Each bottleneck vertex is part of a unique bottleneck path. 
    For each bottleneck path $u_1, \dots, u_k$ we always have that $(u_{i-1}, u_i) \in T$. For all edge weightings $w$, we have $d(s, u_i) = d(s, u_{i-1}) + w((u_{i-1}, u_i))$.
\end{observation}

\subparagraph{Dijkstra with compressed bottleneck paths.} We now modify \cref{algo:dijkstra}.
As a preprocessing step, we run breath-first-search on the unweighted graph $G$ (thereby performing no edge weights comparisons). We compute the levels $L_i$
and store each bottleneck path in an array. For every bottleneck path $u_1, \dots, u_k$,
we replace any outgoing edge $(u_i, v)$, by an edge $(u_1, v)$ of equivalent
weight $w'(u_1, v) = w(u_1, u_2) + \dots + w(u_{i-1}, u_i) + w(u_i, v) $.

After this preprocessing, we run Dijkstra's algorithm on the input $(G', s, w')$. 
When the modified algorithm pops a vertex $u_i$ in a bottleneck path $u_1, \dots, u_k$, it does not immediately push $u_{i+1}$.
Instead, let $w$ be the new minimum in the heap.
Starting from $i$, we perform exponential search over the array that stores the bottleneck path $u_1, \dots, u_k$ to find the maximum $j$ with
$d(s, u_j) < w$. We replace $u_i$ with $u_j$.
This way, we can simulate the sequence \textsc{pop}$(u_i)$, \textsc{push}$(u_{i+1})$, \textsc{pop}$(u_{i+1})$, \textsc{push}$(u_{i+2})$, \dots, \textsc{pop}$(u_j)$
that would be performed by \cref{algo:dijkstra} using $O(1 + \log(j-i))$ edge weight comparisons.
This yields \cref{algo:dijkstra_compressed}.

\begin{algorithm}
    \begin{algorithmic}
        \Require (Di)graph $G = (V, E)$, weighting $w$, designated source $s \in V$
        \State Run BFS on the unweighted $G$ starting at $s$, computing all bottleneck paths.
        \For{Every bottleneck path $u_1, \dots u_k$}
            \State Replace any edge $(u_i, v)$ with $i \ge 2$ by an equivalent edge $(u_1, v)$
        \EndFor
        \State Initialize H and $\distance[]$ as in \cref{algo:dijkstra}
        \While{$H$ is not empty}
            \State u $\gets$ H.\textsc{Pop}()
            \State \textsc{Output}.append($u$)
            \If{$u$ is a vertex $u_i$ of a bottleneck path $u_1, \dots, u_k$}
                \State w $\gets$ H.\textsc{get-min}$()$
                \State Exponential search from $i$ to the maximal $j$ with $\distance[u_j] < \distance[w]$
                \State \textsc{Output}.append($u_{i+1}$, $\ldots$ , $u_{j-1}$, $u_j$)
                \State u $\gets u_j$
            \EndIf
            \State Process the edges $(u, v)$ as in \cref{algo:dijkstra}
        \EndWhile
    \end{algorithmic}
    \caption{Dijkstra's algorithm with compressed bottleneck paths.}
    \label{algo:dijkstra_compressed}
\end{algorithm}

\subparagraph{Analysis.} For the analysis, we will consider the timestamp intervals $\{[a_i, b_i]\}_{i \in [n]}$
as defined by the unmodified \cref{algo:dijkstra}. 
Recall that by Lemma~\ref{lemm:points_define_linearizations}, any choice of $n$ distinct real values $\{ r_i \}$ with $r_i \in [a_i, b_i]$ corresponds to a linearization of $(G, s)$.
Section~\ref{sec:universal}, we used a result by Cardinal, Fiorini, Joret, Jungers, and Munro who note that if $|Z|$ denotes the number of distinct linearizations that can be obtained in this manner then $ \sum\limits_{i \in [n]} \log( b_i - a_i) \in O(n + \log |Z|)$.
We observe that recently, van der Hoog, Rotenberg and Rutchmann~\cite{Hoog2025Simpler} observe a tighter upper bound for both the number of bottleneck vertices and the sum of interval sizes:

\begin{lemma} [Lemma-2.1 in~\cite{Hoog2025Simpler}, a strengthening of Lemma~5 in \cite{cardinal2010sorting}]\label{lemm:long_path_lower_bound}
    Let $G$ be an unweighted graph with $n+1$ vertices. Consider performing BFS over $G$ from a designated source vertex $s$ and suppose that this procedure partitions the vertices of $G$ into $d$ levels. Then $(n-d) \in O(\log \ell(G, s))$.
\end{lemma}
\begin{proof}
    For completeness, we paraphrase the proof from~\cite{Hoog2025Simpler}.
    Since each level $L_i$ must contain at least one vertex, and the levels partition all $n$ vertices, it follows that $\sum_{i = 0}^{d-1} (|L_i| - 1 ) =  n - d$. 

    Consider any permutation $\pi$ of $V \setminus \{ s \}$ in which the elements in $L_0$ appear first (in any order),
    followed by the elements in $L_1$ (in any order), etc.
    Since elements from $L_{i}$ can only be reached through elements of $L_{i-1}$, we can trivially select the edge weights such that $\pi$ is a linearization of $(G, s)$. 
    The number of such permutations is
    \[
        \prod_{i=0}^{d-1} |L_i|! \ge \prod_{i=0}^{d-1} 2^{|L_i|-1} = 2^{n-d}  \qedhere
    \]
\end{proof}

\noindent
The second upper bound already appeared in~\cite{cardinal2010sorting} and was paraphrased in~\cite{Hoog2019Preprocessing} and~\cite{Hoog2025Simpler}.
We use the latter formulation since it is closest to our current domain: 

\begin{lemma}[Lemma~2.3 in~\cite{Hoog2025Simpler}] \label{lemm:intervals_to_linearizations_better}
    Let $\ell(G, s)$ denote the number of linearizations of $(G, s)$. Then
    \[
        \sum_{x_i \in V \setminus \{ s \}} \log(b_i - a_i) \in O(\log \ell(G, s)).
    \]
\end{lemma}

\begin{theorem} \label{theo:dijkstra_compressed_comparisons}
    \cref{algo:dijkstra_compressed} performs $O(m_f + \log \ell(G, s))$ comparisons.
\end{theorem}
\begin{proof}
    We first show that the number of while loop iterations is $O(\log \ell(G, s))$.
    This accounts for doing $O(1)$ comparisons per iteration.
    Let $\nb$ be the number of non-bottleneck vertices.
    Among three consecutive iterations of the while loop in \cref{algo:dijkstra_compressed},
    at least one iteration pops a non-bottleneck vertex.
    (The first iteration might pop $u_1$, the second one pops $u_2$ and does an exponential search,
    but then the third iteration pops the vertex $w$.)
    Thus, the number of while loop iterations, push operations and pop operations is $O(\nb)$.
    There are at least $n-d$ and at most $2(n-d)$ non-bottleneck vertices,
    so \cref{lemm:long_path_lower_bound} shows $\nb \in O(\log \ell(G, s))$.
    The number of comparisons performed during these \textsc{Pops} is $O(\sum\limits_{x_i \textnormal{ is popped} } (1 + \log(b_i - a_i))) = O(\overline{n} + \sum\limits_{x_i \in V \setminus \{ v \} } \log(b_i - a_i))$.
    By \cref{lemm:long_path_lower_bound} and \cref{lemm:intervals_to_linearizations_better}, this is in $O(\log \ell(G, s))$. 

    Next, let us bound the comparisons performed in the exponential search.
    The exponential search in \cref{algo:dijkstra_compressed} starts from some $i$'th vertex in a bottleneck path and returns the $j$'th vertex. 
    This procedure uses $O(1 + \log(j-i))$ comparisons. As before, we may charge the $O(1)$ term to the previous non-bottleneck vertex that was popped, charging $O(\overline{n}) \subset O(\log \ell(G, s))$ in total. 
    In \cref{algo:dijkstra}, the vertex $w$ was pushed before $u_i$, and will be popped after $u_j$.
    If we denote by $[a_w, b_w]$ the timestamp interval of $w$ then it follows that 
    $j - i \le b_w - a_w$. Each exponential search involves a different vertex $w$, 
    since $w$ will be popped in the next while loop iteration.
    Hence, the number of comparisons of the second term is at most $O(\sum\limits_{x_i \in V \backslash{ s}} \log(b_i - a_i))$,
    which is $O(\log \ell(G, s))$.
\end{proof}

\begin{theorem} \label{theo:dijkstra_comparison_lower_bound}
    Let $m_f$ be the number of forward edges for some set of edge weights $w$.
    Then, for any linear decision tree for the distance ordering problem,
    there exists a (possibly different) set of edge weights $\ol{w}$
    for which the decision tree performs $w(m_f)$ comparisons.
\end{theorem}
\begin{proof}
    Given $w$, construct the exploration tree and the set of forward edges.
    For every edge $e$ in the exploration tree, put $\ol{w}(e) = w(e)$.
    Without loss of generality, suppose these weights are in $[1, M]$.
    All other non-forward edges get an arbitrary weight in $[1, M]$.
    Let $U$ be the largest coefficient used in a linear comparison performed by the linear decision tree. 

    Let $e_1, \dots, e_{m_f}$ be the forward edges in a uniformly random order.
    We put $\ol{w}(e_1) = 1 / n$ and $\ol{w}(e_i) = M U n^2 \cdot 2^i$.
    This choice of weights has the following two properties:
    If we know the distance ordering for $\ol{w}$, then we know which edge is $e_1$.
    If we perform some comparison that contains one or more $e_j$,
    the outcome of this comparisons only depends on the sign of the coefficients of
    the highest-indexed $e_j$.
    Thus, any decision tree that computes the distance ordering
    induces a strategy for following game:
    Given an array that is a permutation of $[m_f]$, you can query any subset
    of the array and get the index of the maximum.
    The goal is to find the index of the element $1$.
    Any deterministic strategy for this game needs to do $m_f-1$ queries.
    Any randomized strategy for this game with success probability
    $1/2$ needs to do $(m_f - 1)/2$ queries.
\end{proof}

\begin{theorem}
    \cref{algo:dijkstra_compressed} is universally optimal with respect to the number of edge weight comparisons made by the algorithm.
\end{theorem}
\begin{proof}
    By \cref{theo:dijkstra_comparison_lower_bound}, \cref{algo:dijkstra_compressed}
    performs $O(m_f + \log \ell(G, s))$ comparisons.
    By an information theoretical lower bound, any algorithm needs to do $w(\log e(G, s))$ comparisons.
    By \cref{theo:dijkstra_compressed_comparisons}, any algorithm needs to do $w(m_f)$ comparisons.
\end{proof}

\end{document}